\documentclass[onecolumn,journal]{IEEEtran}
\ifCLASSINFOpdf
\else
\fi

\usepackage[T1]{fontenc}
\usepackage[latin9]{inputenc}
\usepackage{color}
\usepackage{array}
\usepackage{float}
\usepackage{mathrsfs}
\usepackage{mathtools}
\usepackage{amsthm}
\usepackage{amstext}
\usepackage{amssymb}
\usepackage{stmaryrd}
\usepackage{graphicx}
\usepackage{pgfplots}

\makeatletter

\theoremstyle{plain}
\newtheorem{theorem}{\protect\theoremname}
\theoremstyle{plain}

\theoremstyle{plain}
\newtheorem{corollary}[theorem]{\protect\corollaryname}
\theoremstyle{plain}
\newtheorem{lemma}[theorem]{\protect\lemmaname}
\theoremstyle{definition}
\newtheorem{example}[theorem]{\protect\examplename}
\theoremstyle{definition}

\AtBeginDocument{
	
}

\makeatother

  \providecommand{\corollaryname}{Corollary}
  \providecommand{\examplename}{Example}
  \providecommand{\lemmaname}{Lemma}
  \providecommand{\propositionname}{Proposition}
  \providecommand{\theoremname}{Theorem}
  \providecommand{\definitionname}{Definition}

\begin{document}

%
\title{Complete Weight Enumerators of Some Linear Codes}
%
%
%

\author{Shudi Yang    and
        Zheng-An Yao    
	
\thanks{S.D. Yang is with the
              Department of Mathematics,
Sun Yat-sen University, Guangzhou 510275 and School of Mathematical
Sciences, Qufu Normal University, Shandong 273165, P.R.China.\protect\\

Z.-A. Yao is with the
               Department of Mathematics,
Sun Yat-sen University, Guangzhou 510275, P.R.China.\protect\\
	\protect\\
	E-mail: yangshd3@mail2.sysu.edu.cn,~{mcsyao@mail.sysu.edu.cn}}
\thanks{Manuscript received *********; revised ********.}
}

%
%

\markboth{Journal of \LaTeX\ Class Files,~Vol.~13, No.~9, September~2014}%
{Shell \MakeLowercase{\textit{et al.}}: Bare Demo of IEEEtran.cls for Journals}
%



\maketitle

\begin{abstract}
Linear codes have been an interesting topic in both theory and practice for many years.
In this paper, for an odd prime $p$, we determine the explicit complete weight enumerators of two classes of linear codes over $\mathbb{F}_p$ and they may have applications in cryptography and secret sharing schemes. Moreover, some examples are included to illustrate our results.
\end{abstract}

\begin{IEEEkeywords}
Linear code, complete weight enumerator, quadratic form, Gauss sum.
\end{IEEEkeywords}

%
\IEEEpeerreviewmaketitle

\section{Introduction}\label{sec:intro}
%
%
%
%
\IEEEPARstart{T}{hroughout} this paper, let $p$ be an odd prime. Denote by $\mathbb{F}_p$ a
finite field with $p$ elements. An $[n, \kappa, \delta]$ linear code
$C$ over $\mathbb{F}_p$ is a $\kappa$-dimensional subspace of
$\mathbb{F}_p^n$ with minimum distance $\delta$ \cite{macwilliams1977theory}.

Let $A_i$ denote the number of codewords with Hamming weight
$i$ in a linear code $C$ of length $n$. The (ordinary) weight enumerator of  $C$  is defined by
$$A_0+A_1z+A_2z^2+\cdots+A_nz^n,$$
where $A_0=1$. The sequence $(A_0,A_1,A_2,\cdots,A_n)$ is called the (ordinary) weight
distribution of the code $C$.

The complete weight enumerator of a code $C$ over $\mathbb{F}_p$ enumerates the codewords according to the number of symbols of each kind contained in each codeword. Denote the
field elements by $\mathbb{F}_p=\{w_0,w_1,\cdots,w_{p-1}\}$, where $w_0=0$.
Also let $\mathbb{F}_p^*$ denote $\mathbb{F}_p\backslash\{0\}$.
For a codeword $\mathsf{c}=(c_0,c_1,\cdots,c_{n-1})\in \mathbb{F}_p^n$, let $w[\mathsf{c}]$ be the
complete weight enumerator of $\mathsf{c}$ defined as
$$w[\mathsf{c}]=w_0^{k_0}w_1^{k_1}\cdots w_{p-1}^{k_{p-1}},$$
where $k_j$ is the number of components of $\mathsf{c}$ equal to $w_j$, $\sum_{j=0}^{p-1}k_j=n$.
The complete weight enumerator of the code $C$ is then
$$\mathrm{CWE}(C)=\sum_{\mathsf{c}\in C}w[\mathsf{c}].$$

The weight distribution of a linear code has attracted a lot of interest for many years and we
refer the reader to \cite{Zheng2015weight,ding2011,ding2013hamming,dinh2015recent,feng2008weight,li2014weight,luo2008weight,sharma2012weight,vega2012weight,wang2012weight,yuan2006weight,zheng2013weight} and references therein for an overview of the related researches. It is not difficult to see that the complete weight enumerators are just the (ordinary) weight
enumerators for binary linear codes. While for nonbinary linear codes, the weight enumerators can be obtained from their complete weight enumerators.

 The information of the complete weight enumerator of a linear code is of vital use both in theories and in practical applications. For instance, Blake and Kith investigated the complete weight enumerator of Reed-Solomon codes and showed that they could be helpful in soft decision decoding~\cite{Blake1991,kith1989complete}. In~\cite{helleseth2006}, the study of the monomial and quadratic bent functions was related to the complete weight enumerators of linear codes. It was illustrated by Ding $et~al.$~\cite{ding2007generic,Ding2005auth} that the complete weight enumerator can be applied to calculate the deception probabilities of certain authentication codes. In~\cite{chu2006constant,ding2008optimal,ding2006construction}, the authors studied the complete weight enumerators of some constant
composition codes and presented some families of optimal constant composition codes.

However, it is usually an extremely difficult problem to evaluate the complete
weight enumerator of linear codes and there are few information on this topic in literature besides the above mentioned~\cite{Blake1991,kith1989complete,chu2006constant,ding2008optimal,ding2006construction}.
Kuzmin and Nechaev considered the
generalized Kerdock code and related linear codes over Galois rings and determined their complete weight enumerators in~\cite{kuzmin1999complete} and \cite{kuzmin2001complete}. Very recently, Li, Yue and Fu~\cite{li2015complete} obtained the complete weight enumerators of some cyclic codes by using Gauss sums. In this paper, we shall determine the complete weight enumerators of a class of linear codes over finite fields.

Let $\bar{D}=\{d_1,d_2,\cdots,d_n\}\subseteq \mathbb{F}_{p^m}$. Denote by $\mathrm{Tr}$ the trace function from $\mathbb{F}_{p^m}$ to $\mathbb{F}_{p}$. A linear code associated with $\bar{D}$ is defined by
\begin{equation*}\label{def:CD'}
    C_{\bar{D}}=\{(\mathrm{Tr}(ad_1),\mathrm{Tr}(ad_2),\cdots,\mathrm{Tr}(ad_n)):
       a\in \mathbb{F}_{p^m}\},
\end{equation*} and $\bar{D}$ is called the defining set of this code $C_{\bar{D}}$~(see \cite{dingkelan2014binary,ding2015twodesign,ding2015twothree} for details).

It should be noted that the authors in~\cite{dingkelan2014binary,ding2015twodesign} and~\cite{ding2015twothree} gave the definition of the code $C_{\bar{D}}$ and the defining set $\bar{D}$. The authors in~\cite{dingkelan2014binary} established binary linear codes $C_{\bar{D}}$ with three weights. In~\cite{ding2015twodesign}, Ding presented the general construction of the linear codes and determined their weights especially for three special codes. The authors in~\cite{ding2015twothree} presented the defining set $\bar{D}=\{x\in \mathbb{F}_{p^m}^*:\mathrm{Tr}(x^{2})=0\}$ to construct a class of linear codes $C_{\bar{D}}$ with two and three nonzero weights and investigated their application in secret sharing.

In this paper, the defining set $D$ of the code $C_D$ is given by
\begin{equation}\label{def:D}
     D=\{x\in \mathbb{F}_{p^m}^*:\mathrm{Tr}(x^{2d})=0\}=\{d_1,d_2,\cdots,d_n\}
\end{equation}
for an integer $d$ coprime to $(p^m-1)/2$, i.e., $\mathrm{gcd}(d,(p^m-1)/2)=1$.
Let \begin{equation}\label{def:CD}
    C_{D}=\{(\mathrm{Tr}(ad_1),\mathrm{Tr}(ad_2),\cdots,\mathrm{Tr}(ad_n)):
       a\in \mathbb{F}_{p^m}\}.
\end{equation}

Note that $\mathrm{gcd}(d,(p^m-1)/2)=1$ leads to
$$\{x^{2d}:x\in \mathbb{F}_{p^m}^*\}=\{x^{2}:x\in \mathbb{F}_{p^m}^*\},$$
which means that
\begin{eqnarray*}
    D&=&\{x\in \mathbb{F}_{p^m}^*:\mathrm{Tr}(x^{2d})=0\}\\
     &=&\{x\in \mathbb{F}_{p^m}^*:\mathrm{Tr}(x^{2})=0\}=\bar{D}.
\end{eqnarray*}
Therefore the codes $C_D$ of \eqref{def:CD} and $C_{\bar{D}}$ depicted in~\cite{ding2015twothree} are exactly the same code. Thus we only focus on the defining set $$ \bar{D}=\{x\in \mathbb{F}_{p^m}^*:\mathrm{Tr}(x^{2})=0\}=\{d_1,d_2,\cdots,d_n\}$$ in the sequel and we denote it by $D$ for convenience.

More naturally, a generalization of the code $C_D$ is given by
\begin{equation}\label{def:CDb}
    C_{D,b}=\{(\mathrm{Tr}(ad_1)+b,\mathrm{Tr}(ad_2)+b,\cdots,\mathrm{Tr}(ad_n)+b):a\in \mathbb{ F}_{p^m},b\in \mathbb{F}_p\}.
\end{equation}

We will study the complete weight enumerators of $C_D$ and the generalized code $C_{D,b}$, and then their weight enumerators as well. As it turns out that, $C_{D,b}$ is a linear code with five and seven nonzero weights, while the code $C_D$ is a linear code with two and three nonzero weights as was shown in~\cite{ding2015twothree}. This means that the two classes
of linear codes may be of use in cryptography~\cite{mceliece1978public} and secret sharing schemes~\cite{carlet2005linear}. We should mention that the main idea of solving the complete weight enumerators of $C_D$ and $C_{D,b}$ indeed comes from~\cite{ding2015twodesign,ding2015twothree} which were quite inspiring and very well-written and we will employ some results of~\cite{ding2015twothree} in the consequence sections.

The main results of this paper are given below.

\begin{theorem}\label{thcwe:CD}
Let $D$ and $C_D$ be defined as above.

(A) If $m \geqslant 3$ is odd, then the code
$C_D$ is a $[p^{m-1}-1,m]$ linear code over $\mathbb{F}_{p}$ with the
complete weight enumerator
\begin{eqnarray*}
    \mathrm{CWE}(C_D)&=&w_0^{p^{m-1}-1}+(p^{m-1}-1)w_0^{p^{m-2}-1}\prod_{\rho=1}^{p-1}w_{\rho}^{p^{m-2}}\\
&&+\frac{p-1}{2}\left(p^{m-1}+p^{\frac{m-1}{2}}\right)w_0^{p^{m-2}-1+(p-1)p^{\frac{m-3}{2}}}\prod_{\rho=1}^{p-1}w_{\rho}^{p^{m-2}-p^{\frac{m-3}{2}}}\\
&&+\frac{p-1}{2}\left(p^{m-1}-p^{\frac{m-1}{2}}\right)w_0^{p^{m-2}-1-(p-1)p^{\frac{m-3}{2}}}\prod_{\rho=1}^{p-1}w_{\rho}^{p^{m-2}+p^{\frac{m-3}{2}}}\\
\end{eqnarray*}

(B) If $m \geqslant2$ is even, then the code
$C_D$ over $\mathbb{F}_{p}$ has parameters
$$[ p^{m-1}-1-(-1)^{\frac{m}{ 2}(\frac{p-1}{2})^2}(p-1)p^{\frac{m-2}{2}},m]$$ and the
complete weight enumerator
\begin{eqnarray*}
&&\mathrm{C}\mathrm{WE}(C_D)\\
   &&=w_0^{p^{m-1}-1-(-1)^{\frac{m}{2}(\frac{p-1}{2})^2}(p-1)p^{\frac{m-2}{2}}}+\\
&&~~\left(\!p^{m\!-\!1}\!-\!1-\!(\!-\!1)^{\frac{m}{2}(\frac{p\!-\!1}{2})^2}(p\!-\!1)
p^{\frac{m\!-\!2}{2}}\!\right)w_0^{p^{m\!-\!2}\!-\!1-
\!(\!-1\!)^{\frac{m}{2}(\frac{p\!-\!1}{2})^2}(p\!-\!1)p^{\frac{m\!-\!2}{2}}}\prod_{\rho=1}^{p-1}w_{\rho}^{p^{m-2}}\!+\\
&&~~(p-1)\left(p^{m-1}\!+\!(-1)^{\frac{m}{2}(\frac{p-1}{2})^2}p^{\frac{m-2}{2}}\right)
w_0^{p^{m-2}-1}\prod_{\rho=1}^{p-1}w_{\rho}^{p^{m-2}-(-1)^{\frac{m}{2}(\frac{p-1}{2})^2}p^{\frac{m-2}{2}}}\\
\end{eqnarray*}

\end{theorem}

\begin{corollary}\label{wd:CD}(See Theorems 1 and 2 of~\cite{ding2015twothree}) With notation as above.

(A) If $m \geqslant 3$ is odd, then $C_D$ has the
weight distribution given in Table \ref{wdtable:Cd m odd}, where $A_i = 0$ for all other
weights $i$ not listed in the table.

\begin{table}[htbp]
\tabcolsep 1mm \caption{The weight distribution of $C_D$ for the case of odd $m$}\label{wdtable:Cd m odd}
\begin{center}\begin{tabular}{|l|l|}
  \hline
  Weight $i$   & Multiplicity $A_i$            \\
  \hline
   0         &  1        \\
  \hline
   $(p-1)(p^{m-2}-p^{\frac{m-3}{2}})$   &$\frac{p-1}{2}(p^{m-1}+p^{\frac{m-1}{2}}) $   \\
  \hline
   $(p-1)p^{m-2}$   &$p^{m-1}-1 $   \\
  \hline
  $(p-1)(p^{m-2}+p^{\frac{m-3}{2}})$   &$\frac{p-1}{2}(p^{m-1}-p^{\frac{m-1}{2}}) $   \\
    \hline
  \end{tabular}
\end{center}
\end{table}

(B) If $m\geqslant 2$ is even, then $C_D$ has the weight distribution given in Table \ref{wdtable:Cd m even}, where $A_i= 0$ for all other weights $i$ not listed in the table.
\begin{table}[htbp]
\tabcolsep 1mm \caption{The weight distribution of $C_D$ for the case of even $m$}\label{wdtable:Cd m even}
\begin{center}\begin{tabular}{|l|l|}
  \hline
  Weight $i$   & Multiplicity $A_i$            \\
  \hline
   0 \,        &  1        \\
  \hline
   $(p-1)p^{m-2}$   &$ p^{m-1}-(-1)^{\frac{m}{2}(\frac{p-1}{2})^2}(p-1)p^{\frac{m-2}{2}}-1 $   \\
  \hline
  $(p-1)(p^{m-2}-(-1)^{\frac{m}{2}(\frac{p-1}{2})^2}p^{\frac{m-2}{2}})$   &$(p-1)(p^{m-1}+(-1)^{\frac{m}{2}(\frac{p-1}{2})^2}p^{\frac{m-2}{2}})$   \\
    \hline
  \end{tabular}
\end{center}
\end{table}
\end{corollary}

\begin{example}
(i) Let $(p,m)=(3,5)$ and $d=2$. Then the code $C_D$ has parameters $[80,5,48]$ with complete
weight enumerator
$$w_0^{80}+90w_0^{32}w_1^{24}w_2^{24}+80w_0^{26}w_1^{27}w_2^{27}+72w_0^{20}w_1^{30}w_2^{30},$$
and weight enumerator
$$1+90z^{48}+80z^{54}+72z^{60}.$$

(ii) Let $(p,m)=(5,4)$ and $d=5$. Then the code $C_D$ has parameters $[104,4,80]$ with complete
weight enumerator
$$w_0^{104}+520w_0^{24}w_1^{20}w_2^{20}w_3^{20}w_4^{20}+104w_0^{4}w_1^{25}w_2^{25}w_3^{25}w_4^{25},$$
and weight enumerator
$$1+520z^{80}+104z^{100}.$$

These results are consistent with numerical computation by Magma.

\end{example}

\begin{theorem}\label{thcwe:CDb}
Let $D$ and $C_{D,b}$ be defined as above.

(A) If $m \geqslant 3$  is odd, then the code
$C_{D,b}$ is a $[p^{m-1}-1,m+1]$ code over $\mathbb{F}_{p}$ with the
complete weight enumerator
\begin{eqnarray*}
    &&\mathrm{CWE}(C_{D,b})\\
    &&=\sum_{b=0}^{p-1}w_{b}^{p^{m-1}-1}+(p^{m-1}-1)
    \sum_{b=0}^{p-1}w_{b}^{p^{m-2}-1}\prod_{\scriptstyle\rho=0 \atop \scriptstyle \rho\neq b}^{p-1}w_{\rho}^{p^{m-2}}+\\
&&~~\frac{p\!-\!1}{2}\left(p^{m\!-\!1}\!+\!p^{\frac{m\!-\!1}{2}}\right)
\sum_{b=0}^{p-1}w_{b}^{p^{m-2}\!-\!1\!+\!(p-1)p^{\frac{m-3}{2}}}\prod_{\scriptstyle\rho=0 \atop \scriptstyle \rho\neq b}^{p-1}w_{\rho}^{p^{m-2}-p^{\frac{m-3}{2}}}+ \\
&&~~\frac{p\!-\!1}{2}\left(p^{m\!-\!1}\!-\!p^{\frac{m\!-\!1}{2}}\right)
\sum_{b=0}^{p-1}w_{b}^{p^{m-2}\!-\!1\!-\!(p-1)p^{\frac{m-3}{2}}}\prod_{\scriptstyle\rho=0 \atop \scriptstyle \rho\neq b}^{p-1}w_{\rho}^{p^{m-2}+p^{\frac{m-3}{2}}}. \\
\end{eqnarray*}

(B) If $m\geqslant2$ is even, then the code
$C_{D,b}$ over $\mathbb{F}_{p}$ has parameters
$$[ p^{m-1}-1-(-1)^{\frac{m}{ 2}(\frac{p-1}{2})^2}(p-1)p^{\frac{m-2}{2}},m+1]$$ and the
complete weight enumerator
\begin{eqnarray*}
   && \mathrm{C}\mathrm{WE}(C_{D,b})\\
&& =\sum_{b=0}^{p-1}w_{b}^{p^{m-1}-1-(-1)^{\frac{m}{2}(\frac{p-1}{2})^2}(p-1)p^{\frac{m-2}{2}}}+\\
&&~~\!\left(\!p^{m\!-\!1}\!\!\!-\!\!1\!-\!(\!-1\!)^{\frac{\!m}{2}(\frac{p\!-\!1}{2})^2}(\!p\!-\!1\!)p^{\frac{m\!-\!2}{2}}\!\right)
\sum_{b=0}^{p-1}w_b^{p^{m\!-\!2}\!-\!1\!-(\!-1\!)^{\frac{\!m}{2}(\frac{p\!-\!1}{2})^2}(\!p\!-\!1\!)p^{\frac{\!m\!-\!2}{2}}}
\prod_{\scriptstyle\rho=0 \atop \scriptstyle \rho\neq b}^{p-1}w_{\rho}^{p^{\!m\!-\!2}}+\\
&&~~(p\!-\!1)\left(\!p^{m\!-\!1}\!+\!(\!-1\!)^{\frac{m}{2}(\frac{p\!-\!1}{2})^2}p^{\frac{m\!-\!2}{2}}\right)
\sum_{b=0}^{p-1}w_b^{p^{m\!-\!2}\!-\!1}\prod_{\scriptstyle\rho=0 \atop \scriptstyle \rho\neq b}^{p-1}w_{\rho}^{p^{m\!-\!2}\!-\!(\!-1\!)^{\frac{m}{2}(\frac{p\!-\!1}{2})^2}p^{\frac{m\!-\!2}{2}}}.\\
\end{eqnarray*}
\end{theorem}

\begin{corollary}\label{wd:CDb} With notation as above.

(A) If $m \geqslant 3$ is odd, then $C_{D,b}$ has the
weight distribution given in Table \ref{wdtable:Cdb m odd}, where $A_i = 0$ for all other
weights $i$ not listed in the table.
\begin{table}[htbp]
\tabcolsep 1mm \caption{The weight distribution of $C_{D,b}$ for the case of odd $m$}\label{wdtable:Cdb m odd}
\begin{center}\begin{tabular}{|l|l|}
  \hline
  Weight $i$                                 & Multiplicity $A_i$            \\
  \hline
   0                                          &  1        \\
  \hline
  $p^{m-1}-1 $                                &$p-1$\\
  \hline
   $(p-1)p^{m-2}$                             &$p^{m-1}-1$\\
  \hline
   $(p-1)p^{m-2}-1 $                          &$(p-1)(p^{m-1}-1)$\\
  \hline
  $(p-1)(p^{m-2}-p^{\frac{m-3}{2}})$          & $\frac{p-1}{2}(p^{m-1}+p^{\frac{m-1}{2}})$\\
  \hline
  $(p-1)p^{m-2}+p^{\frac{m-3}{2}}-1$          &$\frac{(p-1)^2}{2}(p^{m-1}+p^{\frac{m-1}{2}})$\\
  \hline
  $(p-1)(p^{m-2}+p^{\frac{m-3}{2}})$          &$\frac{p-1}{2}(p^{m-1}-p^{\frac{m-1}{2}})$\\
  \hline
  $(p-1)p^{m-2}-p^{\frac{m-3}{2}}-1 $         &$\frac{(p-1)^2}{2}(p^{m-1}-p^{\frac{m-1}{2}})$\\
  \hline
  \end{tabular}
\end{center}
\end{table}

(B) If $m\geqslant 2$ is even, then $C_{D,b}$ has the weight distribution given in Table \ref{wdtable:Cdb m even}, where $A_i= 0$ for all other weights $i$ not listed in the table.
\begin{table}[htbp]
\tabcolsep 1mm \caption{The weight distribution of $C_{D,b}$ for the case of even $m$}\label{wdtable:Cdb m even}
\begin{center}\begin{tabular}{|l|l|}
  \hline
  Weight $i$   & Multiplicity $A_i$            \\
  \hline
   0 \,        &  1        \\
  \hline
  $p^{m-1}-1-(-1)^{\frac{m}{2}(\frac{p-1}{2})^2}(p-1)p^{\frac{m-2}{2}}$
  &$p-1$\\
  \hline
   $(p-1)p^{m-2}$
   &$p^{m-1}-1-(-1)^{\frac{m}{2}(\frac{p-1}{2})^2}(p-1)p^{\frac{m-2}{2}}$\\
  \hline
  $(p-1)(p^{m-2}-(-1)^{\frac{m}{2}(\frac{p-1}{2})^2}p^{\frac{m-2}{2}})-1$ &$(p\!-\!1)(p^{m\!-\!1}\!-\!1\!-\!(-1)^{\frac{m}{2}(\frac{p\!-\!1}{2})^2}(p\!-\!1)p^{\frac{m-2}{2}})$\\
  \hline
  $(p-1)(p^{m-2}-(-1)^{\frac{m}{2}(\frac{p-1}{2})^2}p^{\frac{m-2}{2}})$
  &$(p-1)(p^{m-1}+(-1)^{\frac{m}{2}(\frac{p-1}{2})^2}p^{\frac{m-2}{2}})$\\
  \hline
  $(p-1)p^{m-2}\!-\!(p\!-\!2)(-1)^{\frac{m}{2}(\frac{p\!-\!1}{2})^2}p^{\frac{m\!-\!2}{2}}\!-\!1$
  &$(p-1)^2(p^{m-1}+(-1)^{\frac{m}{2}(\frac{p-1}{2})^2}p^{\frac{m-2}{2}})$\\
  \hline
  \end{tabular}
\end{center}
\end{table}
\end{corollary}

\begin{example}
(i) Let $(p,m)=(3,5)$ and $d=2$. Then the code $C_{D,b}$ has parameters $[80,6,48]$ with complete
weight enumerator
\begin{eqnarray*}
  &w_0^{80}&+90w_0^{32}w_1^{24}w_2^{24}+72w_0^{30}w_1^{30}w_2^{20}+72w_0^{30}w_1^{20}w_2^{30}
+80w_0^{27}w_1^{27}w_2^{26}\\
  &&+80w_0^{27}w_1^{26}w_2^{27}+80w_0^{26}w_1^{27}w_2^{27}
  +90w_0^{24}w_1^{32}w_2^{24}+90w_0^{24}w_1^{24}w_2^{32}\\
  &&+72w_0^{20}w_1^{30}w_2^{30}+w_1^{80}+w_2^{80},
\end{eqnarray*}

and weight enumerator
$$1+90z^{48}+144z^{50}+160z^{53}+80z^{54}+180z^{56}+72z^{60}+2z^{80}.$$

(ii) Let $(p,m)=(3,4)$ and $d=3$. Then the code $C_{D,b}$ has parameters $[20,5,11]$ with complete
weight enumerator
\begin{eqnarray*}
&w_0^{20}&+20w_0^{9}w_1^{9}w_2^{2}+20w_0^{9}w_1^{2}w_2^{9}+60w_0^{8}w_1^{6}w_2^{6}
  +60w_0^{6}w_1^{8}w_2^{6}\\
  &&+60w_0^{6}w_1^{6}w_2^{8}+20w_0^{2}w_1^{9}w_2^{9}+w_1^{20}+w_2^{20},
  \end{eqnarray*}
and weight enumerator
$$1+40z^{11}+60z^{12}+120z^{14}+20z^{18}+2z^{20}.$$

We can check these results by using Magma.
\end{example}

The remainder of this paper is organized as follows. In Section
\ref{sec:mathtool}, we recall some definitions and results on
quadratic forms and Gauss sums over finite fields. Section \ref{sec:proof}
is devoted to the proofs of Theorems \ref{thcwe:CD} and \ref{thcwe:CDb}, respectively.
Section \ref{sec:conclusion} concludes this paper and makes some
remarks on this topic.

\section{Mathematical foundations}\label{sec:mathtool}

We start with quadratic forms over
 finite fields. Let $q$ be a
power of $p$ and $t$ be a positive integer. By identifying
the finite field $\mathbb{F}_{q^t}$ with a $t$-dimensional vector
space $\mathbb{F}^t_{q}$ over $\mathbb{F}_{q}$, a function $f(x)$
from $\mathbb{F}_{q^t}$ to $\mathbb{F}_{q}$ can be regarded as a
$t$-variable polynomial over $\mathbb{F}_{q}$. The function $f(x)$ is called
a quadratic form if it can be written as a
homogeneous polynomial of degree two on  $\mathbb{F}^t_{q}$ as
follows:
$$f(x_1,x_2,\cdots,x_t)=\sum_{1\leqslant i \leqslant j\leqslant t}a_{ij}x_ix_j,~~a_{ij}\in \mathbb{F}_{q}.$$
Here we fix a basis of  $\mathbb{F}^t_{q}$ over  $\mathbb{F}_{q}$
and identify each $x\in \mathbb{F}_{q^t}$ with a vector
$(x_1,x_2,\cdots,x_t)\in\mathbb{F}^t_{q}$.
 The rank of the quadratic form $f(x)$, rank$(f)$, is defined as the
codimension of the $\mathbb{F}_{q}$-vector space
$$W=\{x\in \mathbb{F}_{q^t}|f(x+z)-f(x)-f(z)=0, ~~for~~all~~z\in \mathbb{F}_{q^t}\}.$$
Then $|W|=q^{t-\mathrm{rank}(f)}$.

For a quadratic form $f(x)$ with $t$ variables over $\mathbb{F}_q$,
there exists a symmetric matrix $A$ of order $t$ over $\mathbb{F}_q$
such that $f(x)=XAX'$, where $X=(x_1,x_2,\cdots,x_t)\in
\mathbb{F}^t_q$ and $X'$ denotes the transpose of $X$. It is known
that there exists a nonsingular matrix $B$ over $\mathbb{F}_q$ such
that $BAB'$ is a diagonal matrix. Making a nonsingular linear
substitution $X=YB$ with $Y=(y_1,y_2,\cdots,y_t)\in \mathbb{F}^t_q$,
we have
$$f(x)=Y(BAB')Y'=\sum^r_{i=1}a_iy^2_i,~~~a_i\in \mathbb{F}^*_q,$$
where $r$ is the rank of $f(x)$. The determinant $\mathrm{det}(f)$
of $f(x)$ is defined to be the determinant of $A$, and $f(x)$ is said to be
nondegenerate if $\mathrm{det}(f)\neq0$.
\begin{lemma}(See Theorems 6.26 and 6.27 of \cite{lidl1983finite})\label{lm: solution of quadra form}
Let $f$ be a nondegenerate quadratic form over $\mathbb{F}_q$,
$q=p^t$ for odd prime $p$, in $l$ variables. Define a function
$\upsilon(\cdot)$ over $\mathbb{F}_q$ by $\upsilon(0)=q-1$ and
$\upsilon(\rho)=-1$ for $\rho\in\mathbb{F}^*_q$. Then for
$b\in\mathbb{F}_q$ the number of solutions of the equation
$f(x_1,\cdots,x_l)=b$ is
\begin{eqnarray*}
\left\{\begin{array}{lll}q^{l-1}+\upsilon(b)q^{\frac{l-2}{2}}\eta_t\left((-1)^\frac{l}{2}\mathrm{det}(f)\right),
&&if~~ l~~ is~~ even,\\
q^{l-1}+q^{\frac{l-1}{2}}\eta_t\left((-1)^\frac{l-1}{2}b~
\mathrm{det}(f)\right), &&if~~ l~~ is
~~odd,\\
\end{array}
\right.
\end{eqnarray*}
where $\eta_t$ is the quadratic character of $\mathbb{F}_q$ defined by
\begin{eqnarray*}
\eta_t(x)=\left\{\begin{array}{lll}1,
&&if ~~x ~~is~~ a~~ square~~ in~~ \mathbb{F}_{p^t}^*,\\
-1,&&if ~~x ~~is~~ a~~ nonsquare~~ in~~ \mathbb{F}_{p^t}^*,\\
0,&&if ~~x=0.\\
\end{array}
\right.
\end{eqnarray*}
\end{lemma}

The canonical additive character of $\mathbb{F}_{p^m}$, denoted $\chi$, is given by
\begin{eqnarray*}
\chi(x)=\zeta_p^{\mathrm{Tr}(x)}
\end{eqnarray*}for all $x\in \mathbb{F}_{p^m}$, where $\zeta_p=e^{2\pi\sqrt{-1}/p}$ and $\mathrm{Tr}$ is a trace function from
$\mathbb{F}_{p^m}$ to $\mathbb{F}_{p}$ defined by
$$\mathrm{Tr}(x)=\sum^{m-1}_{i=0}x^{p^i},~~x\in
\mathbb{F}_{p^m}.$$

In what follows, we abbreviate $\eta_m$ as $\eta$ for simplicity. The quadratic Gauss sum $G(\eta,\chi)$ over $\mathbb{F}_{p^m}$ is defined by
\begin{eqnarray*}
G(\eta,\chi)=\sum_{x\in\mathbb{F}_{p^m}^*}\eta(x)\chi(x)=\sum_{x\in\mathbb{F}_{p^m}}\eta(x)\chi(x),
\end{eqnarray*}
and the quadratic Gauss sum $G(\bar{\eta},\bar{\chi})$ over $\mathbb{F}_{p}$ is defined by
\begin{eqnarray*}
G(\bar{\eta},\bar{\chi})=\sum_{x\in\mathbb{F}_{p}^*}\bar{\eta}(x)\bar{\chi}(x)=\sum_{x\in\mathbb{F}_{p}}\bar{\eta}(x)\bar{\chi}(x),
\end{eqnarray*}
where $\bar{\eta}$ and $\bar{\chi}$ are the quadratic and canonical additive characters of $\mathbb{F}_{p}$, respectively.

The lemmas introduced below will play an important role in the
sequel.

\begin{lemma}(See Theorems 5.15 \cite{lidl1983finite})\label{lm:gauss sum}
With the symbols and notation above, we have
\begin{eqnarray}\label{eq:Gausspm}
G(\eta,\chi)=(-1)^{m-1}(\sqrt{-1})^{\frac{(p-1)^2}{4}m}p^{\frac{m}{2}},
\end{eqnarray}
and
\begin{eqnarray}\label{eq:Gaussp}
G(\bar{\eta},\bar{\chi})=(\sqrt{-1})^{\frac{(p-1)^2}{4}}p^{\frac{1}{2}}.
\end{eqnarray}

\end{lemma}

\begin{lemma}(See Theorem 5.33 of~\cite{lidl1983finite})\label{lm:expo sum}
With the symbols and notation above. Let $f(x)=a_2x^2+a_1x+a_0\in \mathbb{F}_{p^m}[x]$ with
$a_2\neq0$. Then
\begin{eqnarray*}\label{eq:expo sum}
\sum_{x\in
\mathbb{F}_{p^m}}\chi(f(x))=\chi(a_0-a_1^2(4a_2)^{-1})\eta(a_2)G(\eta,\chi).
\end{eqnarray*}
\end{lemma}

\begin{lemma}(See Lemma 7 of~\cite{ding2015twothree})\label{le:eta}
If $m\geq 2$ is even, then $\eta(y) = 1$ for each $y \in \mathbb{F}_{p}^*$.
If $m$ is odd, then $\eta(y) = \bar{\eta}(y)$ for each $y \in
\mathbb{F}_{p}$.

\end{lemma}

\begin{lemma}(See Lemma 8 of~\cite{ding2015twothree})\label{lem:yTrace}
We have the following equality:
\begin{equation*}
\sum_{y \in \mathbb{F}_{p}^*}\sum_{x \in \mathbb{F}_{p^m}}
\zeta_p^{y\mathrm{Tr}(x^2)}= \left\{ \begin{array}{lll} 0  & & if~~ m~~ odd, \\
(-1)^{m-1} (-1)^{\frac{m}{ 2}(\frac{p-1}{2})^2}(p-1)p^{\frac{m}{2}}   & &
if~~m~~ even.\\
\end{array}
\right.
\end{equation*}
\end{lemma}

\begin{lemma}(See Lemma 9 of~\cite{ding2015twothree})\label{lem:length}
Let $n_0 = \#\{x \in \mathbb{F}_{p^m} : \mathrm{Tr}(x^2)
= 0\}$. Then
\begin{equation*}
n_0 = \left\{\begin{array} {lll} p^{ m-1}        && if~~ m~~ odd ,\\
 p^{m-1}-(-1)^{\frac{m}{ 2}(\frac{p-1}{2})^2}(p-1)p^{\frac{m-2}{2}} && if~~ m~~ even  .\\
\end{array}
\right.
\end{equation*}
\end{lemma}

\section{The proofs of the main results}\label{sec:proof}

Our task of this section is to prove Theorems \ref{thcwe:CD} and \ref{thcwe:CDb}
depicted in Section \ref{sec:intro}, while Corollary \ref{wd:CDb} follows immediately from Theorem \ref{thcwe:CDb}. In the following, a series of auxiliary results are described and proved for this purpose.

\begin{lemma}\label{lem:sum3}
Let $a\in \mathbb{F}_{p^m}^* $ and $\rho\in \mathbb{F}_{p}$. Then, for $m\geqslant3$ being odd, we have
\begin{eqnarray*}
  &&\sum_{y\in\mathbb{F}_{p}^*}\sum_{z\in\mathbb{F}_{p}^*}\sum_{x\in\mathbb{F}_{p^m}}\zeta_p^{\mathrm{Tr}(yx^2+azx)-z\rho}\\
  &&=\left\{\begin{array}{lll}
   0 &&if~~ \mathrm{Tr}(a^2)=0~~and~~\rho=0,\\
   0 &&if~~ \mathrm{Tr}(a^2)=0~~and~~\rho\neq0,\\
  (-1)^{\frac{m-1}{2}\frac{p-1}{2}}(p-1)p^{\frac{m+1}{2}}\bar{\eta}(\mathrm{Tr}(a^2))
     &&if~~ \mathrm{Tr}(a^2)\neq0~~and~~\rho=0,\\
   -(-1)^{\frac{m-1}{2}\frac{p-1}{2}}p^{\frac{m+1}{2}}\bar{\eta}(\mathrm{Tr}(a^2))
   &&if~~ \mathrm{Tr}(a^2)\neq0~~and~~\rho\neq0,\\
\end{array} \right.
 \end{eqnarray*}
and for $m\geqslant 2$ being even, we have
  \begin{eqnarray*}
  &&\sum_{y\in\mathbb{F}_{p}^*}\sum_{z\in\mathbb{F}_{p}^*}\sum_{x\in\mathbb{F}_{p^m}}\zeta_p^{\mathrm{Tr}(yx^2+azx)-z\rho}\\
  &&=\left\{\begin{array}{lll}
  -(-1)^{\frac{m}{2}(\frac{p-1}{2})^2}(p-1)^2 p^{\frac{m}{2}}&&if~~ \mathrm{Tr}(a^2)=0~~and~~\rho=0,\\
   (-1)^{\frac{m}{2}(\frac{p-1}{2})^2}(p-1) p^{\frac{m}{2}}&&if~~ \mathrm{Tr}(a^2)=0~~and~~\rho\neq0,\\
   (-1)^{\frac{m}{2}(\frac{p-1}{2})^2}(p-1) p^{\frac{m}{2}}&&if~~ \mathrm{Tr}(a^2)\neq0~~and~~\rho=0,\\
   -(-1)^{\frac{m}{2}(\frac{p-1}{2})^2}p^{\frac{m}{2}}&&if~~ \mathrm{Tr}(a^2)\neq0~~and~~\rho\neq0.\\
\end{array} \right.
 \end{eqnarray*}

\end{lemma}
\begin{proof}
It follows from Lemmas \ref{lm:expo sum} and \ref{le:eta} that
\begin{eqnarray*}
  &&\sum_{y\in\mathbb{F}_{p}^*}\sum_{z\in\mathbb{F}_{p}^*}\sum_{x\in\mathbb{F}_{p^m}}\zeta_p^{\mathrm{Tr}(yx^2+azx)-z\rho}\\
  &&=\sum_{z\in\mathbb{F}_{p}^*}\zeta_p^{-z\rho}\sum_{y\in\mathbb{F}_{p}^*}\sum_{x\in\mathbb{F}_{p^m}}\zeta_p^{\mathrm{Tr}(yx^2+azx)}\\
  &&=\sum_{z\in\mathbb{F}_{p}^*}\zeta_p^{-z\rho}\sum_{y\in\mathbb{F}_{p}^*}\chi\left(-\frac{a^2 z^2}{4y}\right)\eta(y)G(\eta,\chi)\\
  &&=G(\eta,\chi)\sum_{z\in\mathbb{F}_{p}^*}\zeta_p^{-z\rho}\sum_{y_1\in\mathbb{F}_{p}^*}\chi(-a^2 z^2 y_1)\eta\left(\frac{1}{4y_1}\right)\\
   &&=G(\eta,\chi)\sum_{z\in\mathbb{F}_{p}^*}\zeta_p^{-z\rho}\sum_{y\in\mathbb{F}_{p}^*}\chi(-a^2 z^2 y)\eta(y)\\
  &&=G(\eta,\chi)\sum_{z\in\mathbb{F}_{p}^*}\zeta_p^{-z\rho}\sum_{y\in\mathbb{F}_{p}^*}\zeta_p^{-z^2\mathrm{Tr}(a^2)y}\eta(y)\\
  &&=\left\{\begin{array}{lll}
  G(\eta,\chi)\sum_{z\in\mathbb{F}_{p}^*}\zeta_p^{-z\rho}\sum_{y\in\mathbb{F}_{p}^*}\eta(y)&&\mathrm{if}~~\mathrm{Tr}(a^2)=0,\\
  G(\eta,\chi)\sum_{z\in\mathbb{F}_{p}^*}\zeta_p^{-z\rho}\sum_{y\in\mathbb{F}_{p}^*}\zeta_p^{-z^2\mathrm{Tr}(a^2)y}\eta(-z^2\mathrm{Tr}(a^2)y)\eta(-\mathrm{Tr}(a^2))&&\mathrm{if}~~\mathrm{Tr}(a^2)\neq0.\\
\end{array} \right.
\end{eqnarray*}

 For the case of $m$ being odd, by Lemma \ref{le:eta}, we have
 \begin{eqnarray*}
  &&\sum_{y\in\mathbb{F}_{p}^*}\sum_{z\in\mathbb{F}_{p}^*}\sum_{x\in\mathbb{F}_{p^m}}\zeta_p^{\mathrm{Tr}(yx^2+azx)-z\rho}\\
  &&=\left\{\begin{array}{lll}
  G(\eta,\chi)\sum_{z\in\mathbb{F}_{p}^*}\zeta_p^{-z\rho}\sum_{y\in\mathbb{F}_{p}^*}\bar{\eta}(y)&&\mathrm{if}~~\mathrm{Tr}(a^2)=0\\
  G(\eta,\chi)\sum_{z\in\mathbb{F}_{p}^*}\zeta_p^{-z\rho}\sum_{y\in\mathbb{F}_{p}^*}\zeta_p^{-z^2\mathrm{Tr}(a^2)y}\bar{\eta}(-z^2\mathrm{Tr}(a^2)y)\bar{\eta}(-\mathrm{Tr}(a^2))
  &&\mathrm{if}~~\mathrm{Tr}(a^2)\neq0\\
\end{array} \right.\\
  &&=\left\{\begin{array}{lll}
  0&&\mathrm{if}~~\mathrm{Tr}(a^2)=0\\
  G(\eta,\chi)\sum_{z\in\mathbb{F}_{p}^*}\zeta_p^{-z\rho}G(\bar{\eta},\bar{\chi})\bar{\eta}(-\mathrm{Tr}(a^2))
  &&\mathrm{if}~~\mathrm{Tr}(a^2)\neq0\\
\end{array} \right.\\
&&=\left\{\begin{array}{lll}
 0 &&\mathrm{if}~~ \mathrm{Tr}(a^2)=0~~\mathrm{and}~~\rho=0,\\
 0 &&\mathrm{if}~~ \mathrm{Tr}(a^2)=0~~\mathrm{and}~~\rho\neq0,\\
 (p-1)G(\eta,\chi)G(\bar{\eta},\bar{\chi})\bar{\eta}(-\mathrm{Tr}(a^2))
 &&\mathrm{if}~~ \mathrm{Tr}(a^2)\neq0~~\mathrm{and}~~\rho=0,\\
 -G(\eta,\chi)G(\bar{\eta},\bar{\chi})\bar{\eta}(-\mathrm{Tr}(a^2))&&\mathrm{if}~~ \mathrm{Tr}(a^2)\neq0~~\mathrm{and}~~\rho\neq0.\\
\end{array} \right.
\end{eqnarray*}
The desired conclusions then follow from Lemma \ref{lm:gauss sum} and the fact that $$(-1)^{\frac{p-1}{2}+\frac{m+1}{2}(\frac{p-1}{2})^2}=(-1)^{\frac{m-1}{2}\frac{p-1}{2}}.$$

 Similarly, for the case of $m$ being even, we can deduce that
 \begin{eqnarray*}
  &&\sum_{y\in\mathbb{F}_{p}^*}\sum_{z\in\mathbb{F}_{p}^*}\sum_{x\in\mathbb{F}_{p^m}}\zeta_p^{\mathrm{Tr}(yx^2+azx)-z\rho}\\
  &&=\left\{\begin{array}{lll}
  (p-1)G(\eta,\chi)\sum_{z\in\mathbb{F}_{p}^*}\zeta_p^{-z\rho}&&\mathrm{if}~~\mathrm{Tr}(a^2)=0\\
  -G(\eta,\chi)\sum_{z\in\mathbb{F}_{p}^*}\zeta_p^{-z\rho}
  &&\mathrm{if}~~\mathrm{Tr}(a^2)\neq0\\
\end{array} \right.\\
 &&=\left\{\begin{array}{lll}
 (p-1)^2G(\eta,\chi) &&\mathrm{if}~~ \mathrm{Tr}(a^2)=0~~\mathrm{and}~~\rho=0,\\
  -(p-1)G(\eta,\chi) &&\mathrm{if}~~ \mathrm{Tr}(a^2)=0~~\mathrm{and}~~\rho\neq0,\\
  -(p-1)G(\eta,\chi) &&\mathrm{if}~~ \mathrm{Tr}(a^2)\neq0~~\mathrm{and}~~\rho=0,\\
  G(\eta,\chi)&&\mathrm{if}~~ \mathrm{Tr}(a^2)\neq0~~\mathrm{and}~~\rho\neq0.\\
\end{array} \right.
\end{eqnarray*}
From Lemma \ref{lm:gauss sum} again, we obtain the desired conclusions.
 \end{proof}

\begin{lemma}\label{lem:Na rho}
For any $a\in \mathbb{F}_{p^m}^* $ and any $\rho\in \mathbb{F}_{p}$, let
\begin{equation*}
    N_a(\rho)=\#\{x\in \mathbb{F}_{p^m}:\mathrm{Tr}(x^2)=0 ~~ \mathrm{and }~~ \mathrm{Tr}(ax)=\rho \}.
\end{equation*}
Then,
 for $m\geqslant 3$ being odd, we have
\begin{equation*}
   N_a(\rho)=\left\{\begin{array}{lll}
   p^{m-2} &&if~~ \mathrm{Tr}(a^2)=0~~and~~\rho=0,\\
   p^{m-2} &&if~~ \mathrm{Tr}(a^2)=0~~and~~\rho\neq0,\\
   p^{m-2}\!+\!(-1)^{\frac{m-1}{2}\frac{p-1}{2}}(p-1)p^{\frac{m-3}{2}}\bar{\eta}(\mathrm{Tr}(a^2))
     &&if~~ \mathrm{Tr}(a^2)\neq0~~and~~\rho=0,\\
   p^{m-2}-(-1)^{\frac{m-1}{2}\frac{p-1}{2}}p^{\frac{m-3}{2}}\bar{\eta}(\mathrm{Tr}(a^2))
   &&if~~ \mathrm{Tr}(a^2)\neq0~~and~~\rho\neq0,\\
\end{array} \right.
 \end{equation*}
 and for $m\geqslant 2$ being even, we have
 \begin{equation*}
   N_a(\rho)=\left\{\begin{array}{lll}
   p^{m-2}-(-1)^{\frac{m}{2}(\frac{p-1}{2})^2}(p-1)p^{\frac{m-2}{2}}&&if~~ \mathrm{Tr}(a^2)=0~~and~~\rho=0,\\
   p^{m-2}&&if~~ \mathrm{Tr}(a^2)=0~~and~~\rho\neq0,\\
   p^{m-2}&&if~~ \mathrm{Tr}(a^2)\neq0~~and~~\rho=0,\\
   p^{m-2}-(-1)^{\frac{m}{2}(\frac{p-1}{2})^2}p^{\frac{m-2}{2}}&&if~~ \mathrm{Tr}(a^2)\neq0~~and~~\rho\neq0.\\
\end{array} \right.
 \end{equation*}
\end{lemma}
\begin{proof}
For any $a\in \mathbb{F}_{p^m}^* $ and any $\rho\in \mathbb{F}_{p}$, we have
\begin{eqnarray*}
    N_a(\rho)&=&
        p^{-2}\sum_{x\in\mathbb{F}_{p^m}}\left(\sum_{y\in\mathbb{F}_{p}}\zeta_p^{y\mathrm{Tr}(x^2)}\right)
        \left(\sum_{z\in\mathbb{F}_{p}}\zeta_p^{z(\mathrm{Tr}(ax)-\rho)}\right)\\
    &=& p^{-2}\sum_{z\in\mathbb{F}_{p}^*}\sum_{x\in\mathbb{F}_{p^m}}\zeta_p^{z(\mathrm{Tr}(ax)-\rho)}
    +p^{-2}\sum_{y\in\mathbb{F}_{p}^*}\sum_{x\in\mathbb{F}_{p^m}}\zeta_p^{y\mathrm{Tr}(x^2)}+\\
    && p^{-2}\sum_{y\in\mathbb{F}_{p}^*}\sum_{z\in\mathbb{F}_{p}^*}\sum_{x\in\mathbb{F}_{p^m}}\zeta_p^{\mathrm{Tr}(yx^2+azx)-z\rho}+p^{m-2}.
 \end{eqnarray*}

 Note that
\begin{equation*}
\sum_{z\in\mathbb{F}_{p}^*}\sum_{x\in\mathbb{F}_{p^m}}\zeta_p^{z(\mathrm{Tr}(ax)-\rho)}=\sum_{z\in\mathbb{F}_{p}^*}\zeta_p^{-z\rho}\sum_{x\in\mathbb{F}_{p^m}}\zeta_p^{z\mathrm{Tr}(ax)}=0
\end{equation*}
since
\begin{equation*}
   \sum_{x\in\mathbb{F}_{p^m}}\zeta_p^{z\mathrm{Tr}(ax)}=\left\{\begin{array}{lll}p^m,&&\mathrm{if}~~z=0,\\
0,&&\mathrm{if}~~z\in\mathbb{F}_{p}^*.\\
\end{array} \right.
 \end{equation*}
The desired conclusions then follow from Lemma \ref{lem:yTrace} and \ref{lem:sum3}.

\end{proof}

\begin{lemma}\label{lem:ti}
Suppose that $m\geqslant 3$ is odd. Let $t_i=\#\{x\in \mathbb{F}_{p^m}^*:\bar{\eta}(\mathrm{Tr}(x^2))=i\}$ with $i\in\{0,1,-1\}$. Then
 \begin{eqnarray*}
 \left\{\begin{array}{lll}
 t_0 &=& p^{m-1}-1, \\
 t_1 &=& \frac{p-1}{2}\left(p^{m-1}+(-1)^{\frac{m-1}{2}\frac{p-1}{2}}p^{\frac{m-1}{2}}\right), \\
 t_{-1} &=& \frac{p-1}{2}\left(p^{m-1}-(-1)^{\frac{m-1}{2}\frac{p-1}{2}}p^{\frac{m-1}{2}}\right). \end{array} \right.
 \end{eqnarray*}
\end{lemma}
\begin{proof}
The case $i=0$ follows from Lemma \ref{lem:length}.

We only give the proof of the case $i=1$ since the other case $i=-1$ is similar.

Note that $\bar{\eta}(\mathrm{Tr}(x^2))=1$ if and only if $\mathrm{Tr}(x^2)=\beta$, where $\beta$ is a quadratic residue over $\mathbb{F}_{p}$.

For each $x\in \mathbb{F}_{p^m}^*$, we can verify that $\mathrm{Tr}(x^2)$ is a quadratic form over $\mathbb{F}_{p}$ with rank $m$.  It follows at once that $\mathrm{Tr}(x^2)=\beta$ can be transformed into the form
\begin{equation}\label{solution}
    \sum_{i=1}^m x_i^2=\beta,
\end{equation} under an orthonormal basis $\{\alpha_1,\alpha_2,\cdots,\alpha_m\}$ of $\mathbb{F}_{p^m}$
over $\mathbb{F}_{p}$ and
$x=\sum_{i=1}^m x_i\alpha_i$ with $x_i\in \mathbb{F}_{p}$.

Since $\mathbb{F}_{p}$ contains $(p-1)/2$ quadratic residues, the desired conclusion then follows from Equation \eqref{solution} and Lemma \ref{lm: solution of quadra form}.
\end{proof}

\subsection{The proof of Theorem \ref{thcwe:CD}}\label{sec:proofthCD}

Recall that
\begin{equation*}
    C_D=\{(\mathrm{Tr}(ad_1),\mathrm{Tr}(ad_2),\cdots,\mathrm{Tr}(ad_n)):a\in \mathbb{F}_{p^m}\},
\end{equation*} where $ D=\{x\in \mathbb{F}_{p^m}^*:\mathrm{Tr}(x^{2})=0\}$.

By Lemma \ref{lem:length}, the length $n$ of $C_D$ is given by
\begin{equation}\label{eq:lengthCD}
n =n_0-1= \left\{\begin{array} {lll} p^{ m-1}-1        && if~~ m~~ odd ,\\
 p^{m-1}-1-(-1)^{\frac{m}{ 2}(\frac{p-1}{2})^2}(p-1)p^{\frac{m-2}{2}} && if~~ m~~ even  .\\
\end{array}
\right.
\end{equation}

Clearly $a=0$ gives the zero codeword and the contribution to the complete
weight enumerator is
$w_0^{n}.$

Assume that $a\in\mathbb{F}_{p^m}^*$ for the rest of the proof. To determine the complete weight enumerator of each codeword $$(\mathrm{Tr}(ad_1),\mathrm{Tr}(ad_2),\cdots,\mathrm{Tr}(ad_n)),$$ we need to consider the number of solutions $x\in \mathbb{F}_{p^m}^*$ satisfying $\mathrm{Tr}(x^2)=0$ and $\mathrm{Tr}(ax)=\rho$ with $\rho\in \mathbb{F}_{p}$, i.e.,
\begin{eqnarray*}
     n_a(\rho)=\#\{x\in \mathbb{F}_{p^m}^*:\mathrm{Tr}(x^2)=0  ~~ \mathrm{and }~~ \mathrm{Tr}(ax)=\rho\}.
\end{eqnarray*}

It is clearly that
\begin{eqnarray*}
     n_a(\rho)=\left\{\begin{array}{lll}N_a(\rho)-1,&&\mathrm{if}~~\rho=0,\\
N_a(\rho),&&\mathrm{if}~~\rho\neq0.\\\end{array} \right.
\end{eqnarray*}

When $m\geqslant 3$ is odd, it follows from Lemma \ref{lem:Na rho} that
\begin{equation*}
   n_a(\rho)=\left\{\begin{array}{lll}
   p^{m-2}-1 &&\mathrm{if}~~ \mathrm{Tr}(a^2)=0~~\mathrm{and}~~\rho=0,\\
   p^{m-2} &&\mathrm{if}~~ \mathrm{Tr}(a^2)=0~~\mathrm{and}~~\rho\neq0,\\
   p^{m\!-\!2}\!\!-\!\!1\!+\!(\!-\!1)^{\frac{m\!-\!1}{2}\frac{p\!-\!1}{2}}(p\!-\!1)p^{\frac{m\!-\!3}{2}}\bar{\eta}(\mathrm{Tr}(\!a^2\!))
     &&\mathrm{if}~~ \mathrm{Tr}(a^2)\neq0~~\mathrm{and}~~\rho=0,\\
   p^{m-2}-(-1)^{\frac{m-1}{2}\frac{p-1}{2}}p^{\frac{m-3}{2}}\bar{\eta}(\mathrm{Tr}(a^2))
   &&\mathrm{if}~~ \mathrm{Tr}(a^2)\neq0~~\mathrm{and}~~\rho\neq0,\\
\end{array} \right.
 \end{equation*}

When $m\geqslant 2$ is even, Lemma \ref{lem:Na rho} shows that
 \begin{equation*}
   n_a(\rho)=\left\{\begin{array}{lll}
   p^{m-2}-1-(-1)^{\frac{m}{2}(\frac{p-1}{2})^2}(p-1)p^{\frac{m-2}{2}}&&\mathrm{if}~~ \mathrm{Tr}(a^2)=0~~\mathrm{and}~~\rho=0\\
   p^{m-2}&&\mathrm{if}~~ \mathrm{Tr}(a^2)=0~~\mathrm{and}~~\rho\neq0\\
   p^{m-2}-1&&\mathrm{if}~~ \mathrm{Tr}(a^2)\neq0~~\mathrm{and}~~\rho=0\\
   p^{m-2}-(-1)^{\frac{m}{2}(\frac{p-1}{2})^2}p^{\frac{m-2}{2}}&&\mathrm{if}~~ \mathrm{Tr}(a^2)\neq0~~\mathrm{and}~~\rho\neq0\\
\end{array} \right.
 \end{equation*}

The desired conclusions of Theorem \ref{thcwe:CD}
then follow from Lemmas \ref{lem:length} and \ref{lem:ti}.

\subsection{The proof of Theorem \ref{thcwe:CDb}}\label{sec:proofthCDb}
Recall that
$$C_{D,b}=\{(\mathrm{Tr}(ad_1)+b,\mathrm{Tr}(ad_2)+b,\cdots,\mathrm{Tr}(ad_n)+b):a\in \mathbb{F}_{p^m},b\in \mathbb{F}_p\},$$where
$ D=\{x\in \mathbb{F}_{p^m}^*:\mathrm{Tr}(x^{2})=0\}$. Obviously, $C_{D,b}$ has the same length $n$ of \eqref{eq:lengthCD} as that of $C_D$.

The complete weight enumerator of $C_{D,b}$ can be explicitly determined by distinguishing the following cases.

\emph{Case 1:} $a=0$ and $b \in \mathbb{F}_p$.

It can be seen that the corresponding codeword of length $n$ contains $b$ in each coordinate position, which contributes to the complete
weight enumerator $w_b^n$ for each $b \in \mathbb{F}_p$. Then the total contribution of such terms to the complete weight enumerator is $$\sum_{b=0}^{p-1}w_b^n.$$

\emph{Case 2:} $a\neq0$ and $b\in \mathbb{F}_p$.

In this case, we consider
\begin{eqnarray*}
     n_{a,b}(\rho+b)&=&\#\{x\in \mathbb{F}_{p^m}^*:\mathrm{Tr}(x^2)=0 ~ \mathrm{and} ~ \mathrm{Tr}(ax)+b=\rho+b\},
\end{eqnarray*} which leads to
$$n_{a,b}(\rho+b)=n_a(\rho).$$

For a fixed $a\in \mathbb{F}_{p^m}^*$, let $(\mathrm{Tr}(ad_1),\mathrm{Tr}(ad_2),\cdots,\mathrm{Tr}(ad_n))$ be a nonzero codeword in $C_D$
with complete weight enumerator
$$w_0^{n_a(0)}w_1^{n_a(1)}\cdots w_{p-1}^{n_a(p-1)}.$$
By the definition of $C_{D,b}$, for a fixed $b\in\mathbb{F}_p$, the corresponding nonzero codeword in $C_{D,b}$ is $$(\mathrm{Tr}(ad_1)+b,\mathrm{Tr}(ad_2)+b,\cdots,\mathrm{Tr}(ad_n)+b),$$
and its contribution to the complete weight enumerator is
$$w_b^{n_a(0)}w_{1+b}^{n_a(1)}\cdots w_{p-1+b}^{n_a(p-1)}.$$
As $b$ runs through $\mathbb{F}_p$, this means that the contributions of such terms to the complete weight enumerator are of the form
$$\sum_{b=0}^{p-1}w_b^{n_a(0)}w_{1+b}^{n_a(1)}\cdots w_{p-1+b}^{n_a(p-1)}.$$

Therefore, the desired conclusions follow from Theorem \ref{thcwe:CD}.

\section{Concluding remarks}\label{sec:conclusion}

In this paper, we proposed the complete weight enumerators of two classes of
the linear codes $C_D$ and $C_{D,b}$ with defining set $D$ for the case of $\mathrm{gcd}(d,(p^m-1)/2)=1$. The ideas of the proofs of Theorems \ref{thcwe:CD} and \ref{thcwe:CDb} came from~\cite{ding2015twodesign,ding2015twothree}.

It should be pointed out that the weight enumerator of $C_D$ was determined in~\cite{ding2015twothree}. And we described the weight enumerator of $C_{D,b}$ which follows
directly from its complete weight enumerator. Some examples were given to confirm our conclusions.


\section{acknowledgements}
The work of Zheng-An Yao is partially supported by the NSFC (Grant No.11271381), the NSFC (Grant No.11431015)
and China 973 Program (Grant No. 2011CB808000).
This work is also partially supported by the NSFC (Grant No.61472457) and Guangdong Natural Science
Foundation (Grant No. 2014A030313161).

\ifCLASSOPTIONcaptionsoff
  \newpage
\fi

\end{document}